\documentclass{llncs}
\usepackage{amsmath,amssymb}
\usepackage[ruled,vlined,linesnumbered,commentsnumbered]{algorithm2e}
\usepackage{graphicx}
\usepackage{subfig}
\usepackage{epstopdf}
\usepackage{float}
\usepackage{tabularx}
\begin{document}
\title{On Robustness in Multilayer Interdependent Network}
\date{\vspace{-5ex}}
\author{\vspace{-5ex}}
\institute{\vspace{-5ex}}
\author{Joydeep Banerjee, Chenyang Zhou, and Arunabha Sen}
\institute{School of Computing, Informatics and Decision System Engineering\\
\small Arizona State University, Tempe, Arizona 85287 \\
\small Email: \{joydeep.banerjee, czhou24, arun.das, asen\}@asu.edu}

\maketitle 
\begin{abstract}
Critical Infrastructures like power and communication networks are highly interdependent on each other for their full functionality. Many significant research have been pursued to model the interdependency and failure analysis of these interdependent networks. However most of these models fail to capture the complex interdependencies that might actually exist between the infrastructures. The \emph{Implicative Interdependency Model} that utilizes Boolean Logic to capture complex interdependencies was recently proposed which overcome the limitations  of the existing models. A number of problems were studies based on this model. In this paper we study the \textit{Robustness} problem in Interdependent Power and Communication Network. The robustness is defined with respect to two parameters $K \in I^{+} \cup \{0\}$ and $\rho \in (0,1]$. We utilized the \emph{Implicative Interdependency Model} model to capture the complex interdependency between the two networks. The model classifies the interdependency relations into four cases. Computational complexity of the problem is analyzed for each of these cases. A polynomial time algorithm is designed for the first case that outputs the optimal solution. All the other cases are proved to be NP-complete. An in-approximability bound is provided for the third case. For the general case we formulate an Integer Linear Program to get the optimal solution and a polynomial time heuristic. The applicability of the heuristic is evaluated using power and communication network data of Maricopa County, Arizona. The experimental results showed that the heuristic almost always produced near optimal value of parameter $K$ for $\rho < 0.42$. 
\vspace{-10pt}
\end{abstract}

\begin{keywords}
Implicative Interdependency Model, Interdependent Networks, Robustness.
\vspace{-15pt}
\end{keywords}

\section{Introduction}
Critical infrastructures (or networks) of a nation are heavily interdependent on each other for their full functionality. Two such infrastructures that engage in a heavy symbiotic dependency are the power and communication networks. For example, consider entities in the power network such as SCADA systems. The primary function of a SCADA system is to control power generation units remotely. This operation relies on the entities of the communication network for receiving control commands. On the other hand communication network entities are dependent on electric power to function properly. The power and communication networks are highly vulnerable to nature induced and man made (terrorist attack) failure. Considering a set of entities in either network failing initially, further failures may be triggered due to the interdependencies between them. 

For analysis of these infrastructures it is imperative to model their interdependencies. A number of such models have been proposed \cite{Bul10}, \cite{Gao11}, \cite{Sha11}, \cite{Ros08}, \cite{Zha05}, \cite{Par13}, \cite{Ngu13}, \cite{Zus11}. However, most of these models fail to account for the complex interdependencies between the networks \cite{Banerjee2014Survey}. In \cite{sen2014identification} the authors described the need to address complex interdependencies as the one in the following example. Consider an entity $a_i$ in power network and entities $b_j,b_k,b_l,b_m,b_n$ in the communication network. For the entity $a_i$ to be operational (i) entities $b_j$  {\em and} $b_k$ needs to be operational, {\em or} (ii) entities $b_k$  {\em and} $b_m$ needs to be operational, {\em or} (iii) entity $b_n$ needs to be operational. The constructed interdependency cannot be represented by graph based interdependency models as described in \cite{Sha11}, \cite{Ros08}, \cite{Zha05}, \cite{Cas13}, \cite{Par13}, \cite{Ngu13}, \cite{Bul10}, \cite{Gao11}. Graph based models fall short in capturing the disjunctive and conjunctive dependencies that exists in real world systems. Authors in \cite{sen2014identification} proposed an \emph{Implicative Interdependency Model} that overcome these limitations. The model uses Boolean Logic to characterize the complex interdependencies. This model was used to study a number of problems on interdependent critical systems \cite{sen2014identification}, \cite{das2014root}, \cite{mazumder2014progressive} and \cite{BanHardening15}.

We use the \emph{Implicative Interdependency Model} (IIM) to study the ``\emph{Robustness}'' problem in interdependent power and communication network. An \textit{Interdependent Network (IDN)} is formally denoted as $\mathcal{I}(A,B,\mathcal{F}(A,B))$. Here $A$ and $B$ are set of entities in power and communication network respectively and function $\mathcal{F}(A,B)$ capturing the interdependency between them through the IIM model (discussed later). ``\emph{Robustness}'' of an interdependent system can be formulated with respect to two parameters $K \in I^{+} \cup \{0\}$ and $\rho \in \mathbb{R}$ with $ 0< \rho \le 1$.  An interdependent system is $(K,\rho)$-robust if a minimum of $K+1$ entities need to fail for failure of at least $\rho (|A|+|B|)$ entities. This robustness value can be treated as a metric to determine the quality of an IDN when it is set up initially. In existing systems this value determines the need and importance of introducing additional measures for improving their robustness. 

The interdependencies using the IIM model can be categorized into four cases (namely case I, II, III and IV). The robustness problem is separately studied for each individual cases. We show that for case I there exists a polynomial time solution to the problem whereas all the other cases are NP-complete. Additionally, we provide an in-approximability bound for case III, an optimal solution using integer linear program for case IV (the general case) and a heuristic for the same. The heuristic is compared with the optimal solution using power and communication network data of Maricopa County, Arizona. From the experimental results we infer that almost always our heuristic produces near optimal solution for $\rho < 0.42$.

The rest of the paper is organized as follows. In Section \ref{IIM_Section} the IIM model is presented, formal definition of the robustness problem and analysis of its computational complexity are done in Section \ref{ProbForm} and \ref{CompAna} respectively, the heuristic and optimal solution to the problem is provided in Section \ref{Solutions} with the corresponding experimental results in Section \ref{ExpRes}, and we conclude the paper is Section \ref{Conc}. 
\vspace{-10pt}

\section{The Implicative Interdependency Model}
\label{IIM_Section}
In this section we describe the IIM model \cite{sen2014identification}. Consider an IDN $\mathcal{I}(A,B,\mathcal{F}(A,B))$. The set $A$ and $B$ consisits of entities $\{a_1,a_2,a_3\}$ and $\{b_1,b_2,b_3,b_4\}$ respectively. The function  $\mathcal{F}(A,B)$ giving the set of dependency equations are represented in Table \ref{tbl:example1idr}. We call the dependency equation of each entity as \textit{Implicative Dependency Relation (IDR)}.  In the given example, an IDR $a_1 \leftarrow b_3 + b_1 b_4$ implies that entity $a_1$ is operational if entity $b_3$ \emph{or} entity $b_1$ \emph{and} $b_4$ are operational. In the IDR each conjunction term e.g. $a_2 a_3$ is referred to as \emph{minterms}.
 
\vspace{-15pt}
\begin{table}[ht!]
\begin{center}
\begin{tabular}{|l||l|}  \hline
{\bf Power Network} & {\bf Comm. Network} \\ \hline
$a_1\leftarrow b_2 + b_4$ & $b_1 \leftarrow a_1 + a_2$ \\ \hline
$a_2 \leftarrow  b_1 b_3$ & $b_2 \leftarrow a_1a_2a_3$ \\ \hline
$a_3 \leftarrow b_3 + b_1 b_4$ & $b_3 \leftarrow a_1 + a_2 a_3$ \\ \hline
$--$ & $b_4 \leftarrow a_2$ \\ \hline
\end{tabular}
\vspace{0.08in}
\caption{Sample Implicative Interdependency Relations of an IDN}
\protect\label{tbl:example1idr}
\end{center}
\vspace{-35pt}
\end{table}

Initial failure of an entity set in $A \cup B$ would cause the failure to cascade until a steady state is reached. The event of an entity failing after the initial failure is termed as \textit{induced failure}. The cascade is assumed to occur in time steps of unit length. Each time step captures the effect of entities killed in all previous time steps. We demonstrate the cascading failure for the interdependent network outlined in Table \ref{tbl:example1idr} through an example. Consider that the entity $a_1$ fails at time step $t=0$. Table \ref{tbl:example1cascade} represents the cascade of failure in each subsequent time steps. In Table \ref{tbl:example1cascade}, for a given entity and time step $'0'$ represents the entity is operational and $'1'$ non operational. In this example a steady state is reached at time step $t=4$ when all entities are non operational. The IIM model also assumes that the dependent entities of all failed entities are killed immediately at the next time step. For example at time step $t=1$ entities $a_1$, $b_2$ and $b_4$ are non operational. Due to the IDR $a_1 \leftarrow b_2+b_4$ entity $a_1$ is killed immediately at time step $t=2$. 

\vspace{-20pt}
\begin{table}[ht!]
\begin{center}
\begin{tabular}{|c|c|c|c|c|c|c|c|}  \hline
\multicolumn{1}{|c|}{Entities} & \multicolumn{7}{c|}{Time Steps ($t$)}\\
\cline{2-8} & $0$ & $1$ & $2$ & $3$ & $4$ & $5$ & $6$ \\\hline \hline
$a_1$ & $0$ & $0$ & $1$ & $1$ & $1$ & $1$ & $1$ \\ \hline
$a_2$ & $1$ & $1$ & $1$ & $1$ & $1$ & $1$ & $1$ \\ \hline
$a_3$ & $0$ & $0$ & $0$ & $0$ & $1$ & $1$ & $1$ \\ \hline
$b_1$ & $0$ & $0$ & $0$ & $1$ & $1$ & $1$ & $1$ \\ \hline
$b_2$ & $0$ & $1$ & $1$ & $1$ & $1$ & $1$ & $1$ \\ \hline
$b_3$ & $0$ & $0$ & $0$ & $1$ & $1$ & $1$ & $1$ \\ \hline
$b_4$ & $0$ & $1$ & $1$ & $1$ & $1$ & $1$ & $1$ \\ \hline
\end{tabular}
\vspace{0.08in}
\caption{Failure cascade propagation when entity $\{a_2\}$ fail at time step $t=0$. A value of $1$ denotes entity failure, and $0$ otherwise}
\protect\label{tbl:example1cascade}
\end{center}
\vspace{-40pt}
\end{table}

The main challenge of the IIM model is accurate formulation of the IDRs. Two possible ways of doing this would be 1) careful analysis of the underlying infrastructures as in \cite{Zus11}, 2) Consultation with domain experts. However we only utilize the IIM model to analyze the Robustness problem and refrain from addressing the mentioned challenge.

\section{Problem Formulation}
\label{ProbForm}
As described before we define the Robustness in Interdependent Network (RIDN) problem with respect to an integer $K \in I^{+} \cup \{0\}$ and a real valued parameter $\rho \in \mathbb{R}$ with $ 0< \rho \le 1$. An IDN $\mathcal{I}(A,B,\mathcal{F}(A,B))$ is $(K,\rho)$ robust if a minimum of $K+1$ entities need to fail initially for a final failure of at least $\rho (|A|+|B|)$ entities. For example, the IDN described in Table \ref{tbl:example1idr} is (0,$\rho$) robust for any value of $\rho \in (0,1]$. This is because initial failure of entity $a_2$ causes all the entities to fail in the steady state. The output of the RIDN problem is the parameter $K$ for a given IDN and a value of $\rho$. We formally state the decision version of the RIDN problem as follows ---.\\ \\
\textbf{The Robustness in Interdependent Network (RIDN) problem}\\ 
\textbf{Instance---} \textit{An IDN $\mathcal{I}(A,B,\mathcal{F}(A,B))$, an integer $K \in I^{+}$ and a real valued parameter $\rho \in \mathbb{R}$ with $ 0< \rho \le 1$.} \\  
\textbf{Decision Version---} \textit{Does there exist a set of entities $S_{I} \subseteq A \cup B$ and $|S_{I}| \le K$ which when failed initially causes a final failure of at least $\rho (|A|+|B|)$ entities.} \\ \\
It is to be noted that a solution to the decision version of RIDN problem would ensure that the IDN is \textbf{not} $(K,\rho)$ robust. We use this notion in our computational complexity proofs.

\section{Computational Complexity Analysis}
\label{CompAna}
The IDRs in the IIM model can be represented in four different forms (1) one minterm of size one, 2) one minterm of arbitrary size, 3) arbitrary number of minterms of size one, and 4) arbitrary number of minterms of arbitrary size (general case). For each of the forms we separately analyze the computational complexity of the RIDN problem. 

\subsection{Case I: Problem Instance with One Minterm of Size One}
The IDRs in the set $\mathcal{F}(A,B)$ have minterms of size 1. With two entities $x_{i}$ and $y_{j}$ of network $A(B)$  and $B(A)$ respectively the IDR $x_{i} \leftarrow y_{j}$ represents this case. Additionally any entity can appear at most once on the left side of the IDR. We provide a polynomial time algorithm (Algorithm \ref{alg:alg1}) and prove its optimality ( Theorem \ref{th:thm1}) that solves the RIDN problem for Case I. 

To develop the algorithm we use the notion of \textit{Kill Set} of an entity $x_i \in A \cup B$ as in \cite{sen2014identification}. \textit{Kill Set} $C_{x_i}$ of an entity $x_i$ are the set of entities failed due to initial failure of $x_i$ alone. Using the concept of \textit{Kill Set} Algorithm \ref{alg:alg1} is developed. 

\begin{algorithm}
\small
	\KwData{An interdependent network $\mathcal{I}(A,B,\mathcal{F}(A,B))$ and a real valued parameter $\rho \in (0,1]$. 
		}
	\KwResult{A set of entities $E$ in $\mathcal{I}(A,B,\mathcal{F}(A,B))$.
		}
	\Begin{			
		  For each entity $x_i \in (A \cup B)$ compute the set of kill sets $\mathcal{C}=\{C_{x_1}, C_{x_2}, ..., C_{x_{|A|+|B|}}\}$, where $C_{x_i} = KillSet(x_i)$ \;
		Initialize $\mathcal{D}=\emptyset$\ and $E=\emptyset$ \;
			\While  {$|\mathcal{D}| < \rho (|A|+|B|$)}{
				Let $x_j$ be the entity having highest $|C_{x_j}|$, in case of a tie choose arbitrarily \;
				Add $x_j$ to set $E$ \;
				Update $\mathcal{D}=\mathcal{D} \cup C_{x_j}$\;
     				\For  {($i=1$; $i \le |A|+|B|$; $i++$)}{
					$C_{x_i}=C_{x_i} \backslash C_{x_j}$ \;
				}
			}
			\Return{$E$} \;
	}		
\caption{RIDN Algorithm for IDNs with Case I type interdependencies}
\label{alg:alg1}
\end{algorithm}

\vspace{-5pt}

\begin{theorem}
\label{th:thm1}
Algorithm \ref{alg:alg1} solves the RIDN problem for Case I optimally in polynomial time.
\vspace{-5pt}
\end{theorem}

\begin{proof}
Computation of \textit{Kill Sets} for all $A \cup B$ entities can be done in $O((|A|+|B|)^3)$ \cite{sen2014identification}. The while loop runs for maximum of $|A|+|B|$ times when $\rho=1$ and \textit{Kill Set} of each entity is only composed of the entity itself. The highest cardinality \textit{Kill Set} among all \textit{Kill Sets} can be found in $O(|A|+|B|)$. The for loop iterates for $|A|+|B|$ times with computation inside it taking $O(|A|+|B|)$ time per iteration. Hence, the time complexity of the while loop in total is $O((|A|+|B|)^3)$. So the overall time complexity of Algorithm \ref{alg:alg1} is $O((|A|+|B|)^3)$. 

We claim that Algorithm \ref{alg:alg1} returns the optimal value of robustness parameter $K=|E|-1$ of an IDN $\mathcal{I}(A,B,\mathcal{F}(A,B))$ with set $E$ containing the minimum number of entities that causes failure of at least $\rho (|A|+|B|)$ entities. The claim is proved by contradiction. Let $E_{OPT}$ be the optimal set that causes failure of at least $\rho (|A|+|B|)$ entities and $x_n$ be an entity in $E_{OPT}\backslash E$. It is proved in \cite{sen2014identification} that in Case I for any two entities $x_i$ and $x_j$, $C_{x_i} \cap C_{x_j} =\emptyset$ or $C_{x_i} \cap C_{x_j} = C_{x_i}$ or $C_{x_i} \cap C_{x_j} = C_{x_j}$ where $x_i \neq x_j$. At any iteration of the while loop the entity $x_j$ with highest cardinality \textit{Kill Set} is selected. Inside the for loop all entities having $C_{x_i} \cap C_{x_j} = C_{x_i}$ and the entity itself would have its \textit{Kill Set} updated to $\emptyset$. Hence the \textit{Kill Set} of the entity $x_n$ would either be set to $\emptyset$ at some iteration of the while loop or didn't have the highest cardinality at any iteration. Hence adding $x_n$ to optimal solution would have made no difference or reduce the number of failed entities. Hence a contradiction. So Algorithm \ref{alg:alg1} returns the optimal number of entities that causes failure of at least $\rho (|A|+|B|)$ entities.

\end{proof}

\subsection{Case II: Problem Instance with One Minterm of Arbitrary Size}
Case II is composed of IDRs having a single minterm of arbitrary size. A minterm of size $p$ with entities $x_{i}$ and $y_{j}$ belonging to network $A(B)$  and $B(A)$ respectively can be represented as $x_{i} \leftarrow \prod^{p}_{j=1} y_{j}$. Thus killing any one entity (or more) from the product term would kill $x_i$. In Theorem \ref{th:thm2} we prove that the decision version of the RIDN problem for Case II is NP complete. 

\begin{theorem}
\label{th:thm2}
The decision version of the RIDN problem for Case II is NP-complete.
\end{theorem}

\begin{proof}
We prove the NP-completeness by giving a transformation from the \textit{Hitting Set Problem}. An instance of the hitting set problem consists of a set of elements $S$ and a set $\mathcal{S}=\{S_1,S_2,S_3,..,S_n\}$ where $S_{i} \subseteq S$, $\forall S_i \in \mathcal{S}$. The question asked in the problem is given an integer $M$ does there exist a set $S' \subseteq S$ with $|S'| \le M$ such that each subset in $\mathcal{S}$ contains at least one element from $S'$. From an instance of the hitting set problem we create an instance of the RIDN problem as follows. For each element $x_i \in S$ we add an entity $b_i \in B$. Similarly for each subset $S_i \in \mathcal{S}$ we add an entity $a_i \in A$. For each subset $S_i=\{x_m,x_n,x_p\}$ (say) we create an IDR $a_i \leftarrow b_m b_n b_p$. The value of $K$ is set to $M$ and $\rho$ is set to $\frac{M+|\mathcal{S}|}{|S|+|\mathcal{S}|}$. It is to be noted that there wont be any cascading failure due to absence of dependency relations of $B$ type entities.

Let there exists a solution to the hitting set problem. So each subset $S_i \in \mathcal{S}$ has at least one element from set $S'$ (with $|S'|=M$). Hence killing the corresponding $B$ type entities from the constructed instance would kill all $A$ type entities. Thus the fraction of entities killed is $\frac{M+|\mathcal{S}|}{|S|+|\mathcal{S}|}=\rho$ solving the RIDN problem. 

On the other way round let there exist a solution to the RIDN problem. It can be shown that the initial failure set would always be chosen from set $B$  to fail $\rho = \frac{M+|\mathcal{S}|}{|S|+|\mathcal{S}|}$ fraction of entities. This is because failure of any $A$ type entity cannot trigger failure of any other entity. Moreover the total number of entities in final failure set is $M+|\mathcal{S}|$ (as $|\mathcal{S}|=|A|$). Thus the failure set must contain all $A$ type entities except for $M$ other entities which has to be chosen from set $B$. So a solution to RIDN problem consisting of entities $B' \subseteq B$ would ensure that for each entity $a_i \in A$ at least one entity in its IDR is killed initially. So the set of elements in $S'$ corresponding to the entities in $B'$ would solve the hitting set problem. Hence proved
\end{proof}

\subsection{Case III: Problem Instance with an Arbitrary Number of Minterm of Size One}
Case III is composed of IDRs having arbitrary number of minterms of size 1. With entities $x_{i}$ and $y_{q}$ belonging to network $A(B)$  and $B(A)$ respectively this case can be represented as $x_{i} \leftarrow \sum^{p}_{q=1} y_{q}$. The given example has $p$ minterms each of size $1$. Thus to kill $x_i$ all entities in its IDR must be killed. In Theorem \ref{th:thm3} we prove that the decision version of the RIDN problem for Case III is NP complete. 

\begin{theorem}
\label{th:thm3}
The decision version of the RIDN problem for Case III is NP-complete.
\end{theorem}

\begin{proof}
We prove that the problem is NP-complete by giving a reduction from the \textit{Densest $p$-Subhypergraph} problem \cite{hajiaghayi2006minimum}, a known NP- complete problem. An instance of the  Densest $p$-Subhypergraph problem includes a hypergraph $G=(V,E)$, a parameter $p$ and a parameter $M$. The problem asks the question whether there exists a set of vertices $|V'| \subseteq V$ and $|V'| \le p$ such that the subgraph induced with this set of vertices has at least $M$ completely covered hyperedges. From an instance of the  Densest $p$-Subhypergraph problem we create an instance of the RIDN problem as follows.  For each vertex $v_i \in V$ we add an entity $b_i \in B$. Similarly for each hyperedge $e_j \in E$ we add an entity $a_j \in A$. For each hyperedge $e_j$ with $e_j= \{v_m, v_n, v_q\}$ (say) an IDR of form $a_j \leftarrow b_m + b_n + b_q$ is created. The value of $K$ is set to $p$ and $\rho$ is set to $\frac{p+M}{|V|+|E|}$. It is to be noted that there wont be any cascading failure due to absence of dependency relations of $B$ type entities.


Let there exist a solution to the Densest $p$-Subhypergraph problem. Then there exist a set $V' \subseteq V$ and $|V'| = p$ that covers completely at least $M$ hypedges in $E$. Thus killing the $B$ type entities corresponding to the vertices in $V'$ would cause at least $M$ $A$ type entities to fail. Hence the fraction of entities killed is $\ge \frac{p+M}{|V|+|E|} =\rho$. So the solution of the Densest $p$-Subhypergraph problem solves the Robustness problem.


For the created instance of the RIDN problem all entities in set $B$ can only fail initially. The $A$ type entities can either fail initially or through induced failure of failing $B$ type entities. Hence initial failure of entities from set $B$ would have the most impact on final number of entities failed. Let us assume that there exists one or many solutions to the RIDN problem. Then at least one solution would have entities only from set $B$. For this solution the number of entities killed on initial failure of $p$ $B$ type entities is at least $p+M$. The additional $M$ entities killed belongs to set $A$. So the vertices in $V$ corresponding to the entities in $B$ would completely cover at least $M$ hyperedges. Thus the solution of RIDN problem solves the Densest $p$-Subhypergraph problem. Hence proved.

\end{proof}

\begin{theorem} {} \label{inapxC2} The RIDN problem for Case III is hard to approximate within a factor $\frac{1}{2^{log(n)^{\lambda}}}$ (where $n=|A \cup B|$) for some $\lambda >0$.
\end{theorem}

\begin{proof}
In \cite{hajiaghayi2006minimum} it is proved the Densest $p$-Subhypergraph problem is hard to approximate within a factor of 
$\frac{1}{2^{log(n)^{\lambda}}}$ with $\lambda >0$. For IDRs of form Case III it is shown in Theorem \ref{th:thm3} that Densest $p$-Subhypergraph problem is a special case of the RIDN problem. Hence proved.
\end{proof}

\subsection{Case IV (General Case): Problem Instance with an Arbitrary Number of Minterms of Arbitrary Size}
Case IV is composed of IDRs having arbitrary number of minterms of arbitrary size. With entities $x_{i}$ and $y_{q}$ belonging to network $A(B)$  and $B(A)$ respectively this case can be represented as $x_{i} \leftarrow \sum^{p}_{j_1=1} \prod^{q_{j_1}}_{j_2=1}y_{j_2}$. The given example has $p$ minterms each of size $q_{j_1}$. In Theorem \ref{th:thm4} we prove that the decision version of the RIDN problem for Case IV is NP complete. 

\begin{theorem}
\label{th:thm4}
The decision version of the RIDN problem for Case IV is NP-complete.
\end{theorem}

\begin{proof}
As IDRs in form of Case II and Case III are subsets of the general case so the RIDN problem for Case IV is NP-complete.
\end{proof}
\section{Solutions to the RIDN Problem}
\label{Solutions}
We propose an optimum solution to the RIDN problem using Integer Linear Programming (ILP) in \ref{sec1}, and a heuristic in section \ref{sec2}
\subsection{Optimal Solution for the RIDN problem}
\label{sec1}
We formulate an ILP that for a given parameter $\rho \in (0,1]$ and an IDN computes the minimum number of entities that need to fail initially for a final failure of $\rho(|A|+|B|)$ entities. Let $K'$ be the solution to the ILP. \textit{Then the IDN is $(K,\rho)$ robust with $K=K'-1$}. The ILP works with two variables $x_{id}$ and $y_{id}$ for each entity $x_i \in A$ and $y_i \in B$ respectively. The parameter $d$ in the variable denotes the time step. $x_{id}$ =1 (or $y_{id}=1$) if at time step $d$ the entity $x_i$ ($y_i$) is in a failed state and $0$ otherwise. With these definitions the objective function can be formulated as follows:
\begin{equation}\label{eqn:ilpobj}
min\overset{m}{\underset{i=1}{\sum}}x_{i0}+\overset{n}{\underset{j=1}{\sum}}y_{j0}
\end{equation}

\noindent
In the above objective function $m$ and $n$ denote the size of the networks $A$ and $B$ respectively. The constraints of the ILP are formally described as follows:\\ 

\noindent
{\em Constraint Set 1:} $x_{id} \geq x_{i(d-1)}, \forall d, 1 \leq d \leq t_f$ and  $y_{id} \geq y_{i(d-1)}, \forall d, 1 \leq d \leq t_f$, where $t_f$ denotes the final time step. The constraint satisfies the property that if the entity $x_i$ fails at time step $d$ it should remain to be in the failed state for all subsequent time steps \cite{sen2014identification}.\\

\noindent
{\em Constraint Set 2}: A brief overview of the constraint set to model the failure propagation through cascades is presented here. Consider an IDR of form ${a_i} \leftarrow {b_j}+{b_k}{b_l} +{b_m}{b_n} {b_q}$. This corresponds to the general case or Case IV as discussed earlier. The constraints created to capture the failure propagation are described in the following steps ---

\vspace{0.05in}
\noindent
{\em Step 1:} We introduce new variables to represent minterms of size greater than one. In this example two new variables $c_1$ and $c_2$ are introduced to represent the minterms ${b_k}{b_l}$ and ${b_m}{b_n} {b_q}$ respectively. This is equivalent of adding two new IDRs $c_1 \leftarrow {b_k}{b_l}$ and $c_2 \leftarrow {b_m}{b_n} {b_q}$ along with the transformed IDR ${a_i} \leftarrow b_j+c_1+c_2$. 

\vspace{0.05in}
\noindent
{\em Step 2:} For each IDR corresponding to the $c$ type variables and untransformed IDRs of form Case II we introduce a linear constraint to capture the failure propagation. For an IDR $c_2 \leftarrow {b_m}{b_n} {b_q}$ the constraint is represented as $c_{2d} \leq y_{m(d-1)} + y_{n(d-1)} + y_{q(d -1)}, \forall d, 1 \leq d \leq t_f$.

\vspace{0.05in}
\noindent
{\em Step 3:} Similarly,  for each transformed IDR and untransformed IDRs of form Case III we introduce a linear constraint to capture the failure propagation. For an IDR ${a_i} \leftarrow b_j + c_1 + c_2$ the constraint is represented as $N \times x_{id} \leq y_{j(d-1)}+ c_{1(d-1)} + c_{2(d-1)}, \forall d, 1 \leq d \leq t_f$. Here $N$ is the number of minterms in the IDR (in this example $N=3$). \\

\noindent
{\em Constraint Set 3}: We must also ensure that at time step $t_f$ at least $\rho (|A|+|B|)$ entities fail. This can be captured by introducing the constraint $\overset{m}{\underset{i=1}{\sum}}x_{i(t_f)}+\overset{n}{\underset{j=1}{\sum}}y_{j(t_f)} \ge \rho(|A|+|B|)$.\\ 
\noindent

So with the objective in (\ref{eqn:ilpobj}) and set of constraints the ILP finds the  minimum number of entities $K'$ which when failed initially causes at least $\rho (|A|+|B|)$ entities to fail at $t_f$. 

\vspace{-10pt}

\subsection{Heuristic Solution for the RIDN problem}
\label{sec2}
A heuristic solution for the RIDN problem is provided in this subsection. Along with the definition of \textit{Kill Set}, we introduce the notion of \textit{Total Minterm Hit Set} of an entity to design the heuristic. Before formal definition of \textit{Total Minterm Hit Set} we first define \textit{Minterm Hit Set} of an entity as follows ---

\begin{definition}
The Minterm Hit Set for an entity $x_j \in A \cup B$ in an interdependent network $\mathcal{I}(A,B,\mathcal{F}(A,B))$ is denoted as $MHS(x_j)$. $MHS(x_j)$ contains the set of all minterms that has the entity $x_j$ in it.
\end{definition} 

\begin{definition}
The Total Minterm Hit Set for an entity $x_j \in A \cup B$ is denoted as $TMHS(x_j)$. It is defined as union of Minterm Hit Set of all entities in $C_{x_j}$ (\textit{Kill Set} of $x_j$).
\end{definition} 

\begin{algorithm}[ht!]
\small
	\KwData{An interdependent network $\mathcal{I}(A,B,\mathcal{F}(A,B))$ and a real valued parameter $\rho \in (0,1]$. 
		}
	\KwResult{An integer $|K_H|-1$ where $K_{H}$ is a set of entities that when killed initially fails at least $\rho (|A|+|B|)$ entities
		}
	\Begin{			
		Initialize $\mathcal{D}=\emptyset$\ and $K_H=\emptyset$ \;
			\While  {$|\mathcal{D}| < \rho (|A|+|B|$)}{
				 For each entity $x_i \in (A \cup B)\backslash D$ compute the kill set $C_{x_i}$ \;
				 For each entity $x_i \in (A \cup B)\backslash D$ compute $TMHS(x_i)$\;
				Let $x_j$ be the entity having highest $|C_{x_j}|$ \;
				\If{There exists multiple entities having highest cardinality Kill Set}{
					Let $x_p$ be an entity having highest $TMHS(x_p)$ with $x_p$ in the set of entities having highest cardinality Kill Set\;
					If there is a tie choose arbitrarily\;
					Add $x_p$ to set $K_H$ \;
					Update $\mathcal{D}=\mathcal{D} \cup C_{x_p}$\;
     				           Update $\mathcal{F}(A,B)$ by removing all IDRs corresponding to entities in $C_{x_p}$ and all minterms in $TMHS(x_p)$\;
				}
				\Else{
				Add $x_j$ to set $K_H$ \;
				Update $\mathcal{D}=\mathcal{D} \cup C_{x_j}$\;
     				Update $\mathcal{F}(A,B)$ by removing all IDRs corresponding to entities in $C_{x_j}$ and all minterms in $TMHS(x_j)$\;		
				}
			}
			\Return{$|K_H|-1$} \;
	}		
\caption{RIDN Algorithm for IDNs with Case I type interdependencies}
\label{alg:alg2}
\end{algorithm} 

Using these definitions a heuristic is formulated in Algorithm \ref{alg:alg2}. For each iteration of the while loop in the algorithm, the operational entity having highest cardinality Kill Set is selected. This ensures that at each step the number of entities failed is maximized. In case of a tie, the entity having highest cardinality Total Minterm Hit Set among the set of tied entities is selected. This causes the selection of the entity that has the potential to kill maximum number of entities in the subsequent steps. Thus, the heuristic greedily minimizes the set of entities which when killed initially fails at least $\rho$ fraction of total entities in the IDN. The heuristic overestimates the parameter $K$ while determining the robustness $(K,\rho)$ of an IDN. The value of the parameter $K$ is equal to $|K_H|-1$ which is the output of Algorithm \ref{alg:alg2}.  Algorithm \ref{alg:alg2} runs in polynomial time, more specifically the run time is $\rho n (n+m)^2$ (where $n=|A|+|B|$ and $m=$ Number of minterms in $\mathcal{F}(A,B)$).

\section{Experimental Results}
\label{ExpRes}
\vspace{-8pt}
We performed experimental comparison between the heuristic and the optimal solution of the RIDN problem. Real world data sets were used for the experiments. The communication network data was obtained from GeoTel (www.geo-tel.com). The dataset contains $2,690$ cell towers, $7,100$ fiber-lit buildings and $42,723$ fiber links of Maricopa County, Arizona, USA. The power network data was obtained from Platts (www.platts.com). It contains $70$ power plants and $470$ transmission lines of the same county. We took four non overlapping regions of the Maricopa county. It is to be noted that the union of the regions does not cover the entire space. The entities of the power and communication network for these four regions were extracted. As per notation, set $A$ and $B$ contain entities of the power and communication network respectively. The number of entities in set $A$ and $B$ are 29 and 19 for Region 1, 29 and 20 for Region 2, 29 and 19 for Region 3, and 33 and 20 for Region 4. The regions were represented by an interdependent network $\mathcal{I}(A,B,\mathcal{F}(A,B))$. For these regions $\mathcal{F}(A,B)$ was generated using the IDR construction rule as defined in \cite{sen2014identification}.

\begin{figure*}[!htb]
	\centering
	\begin{center}
	\subfloat[][Region 1]{\includegraphics[width=5.5cm]{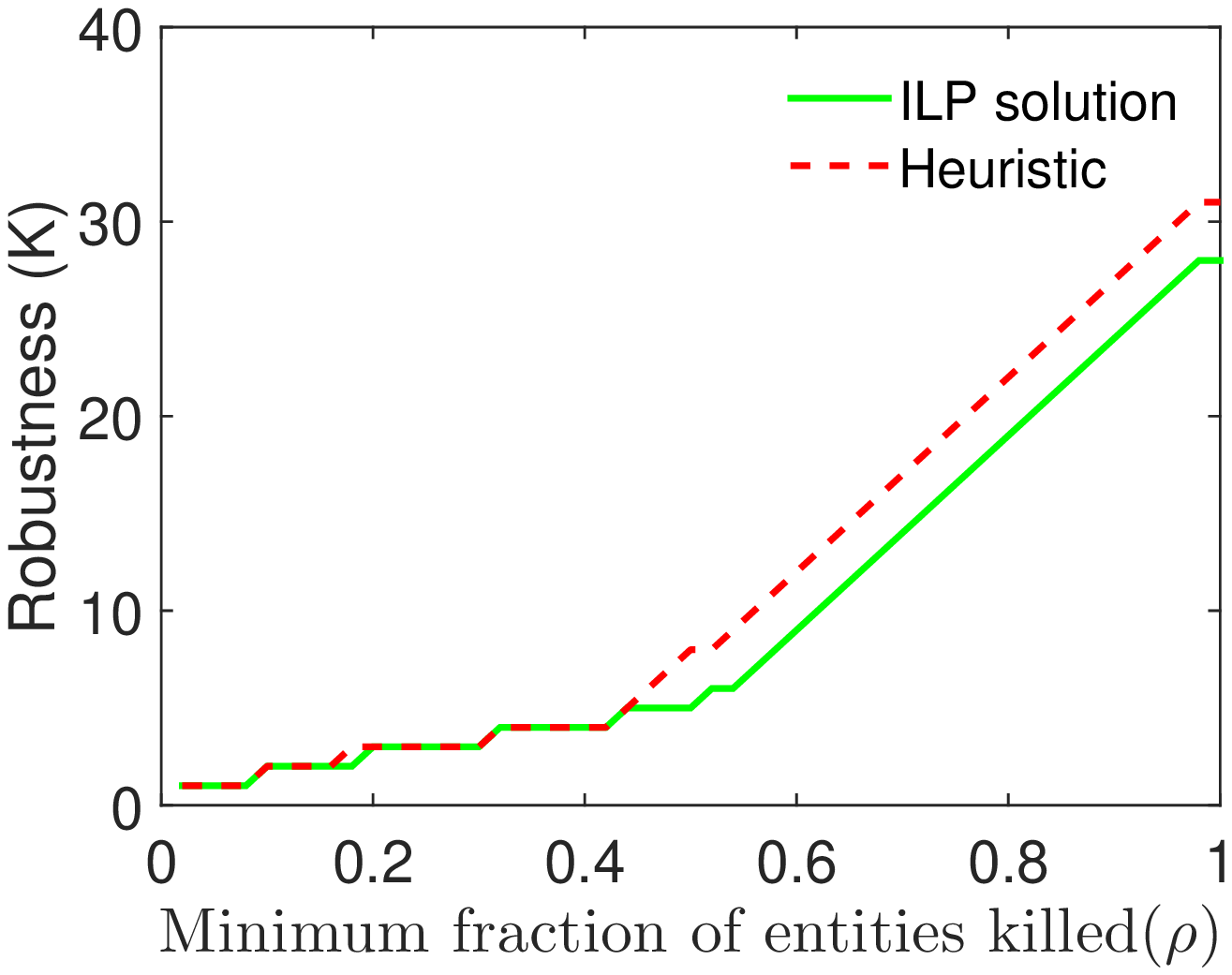}\label{fig:plot1}}
           \subfloat[][Region 2]{\includegraphics[width=5.5cm]{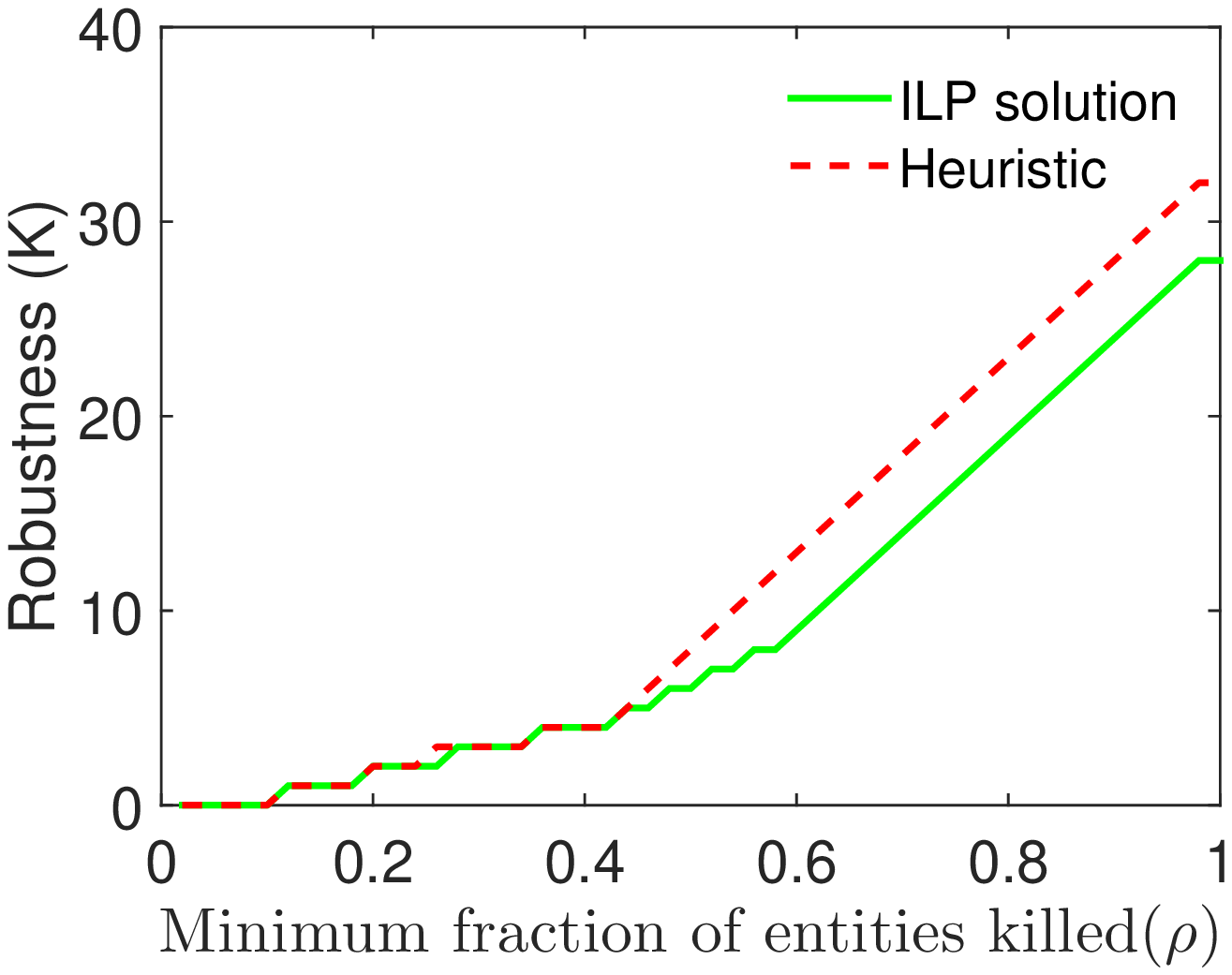}\label{fig:plot2}}\\
           \subfloat[][Region 3]{\includegraphics[width=5.5cm]{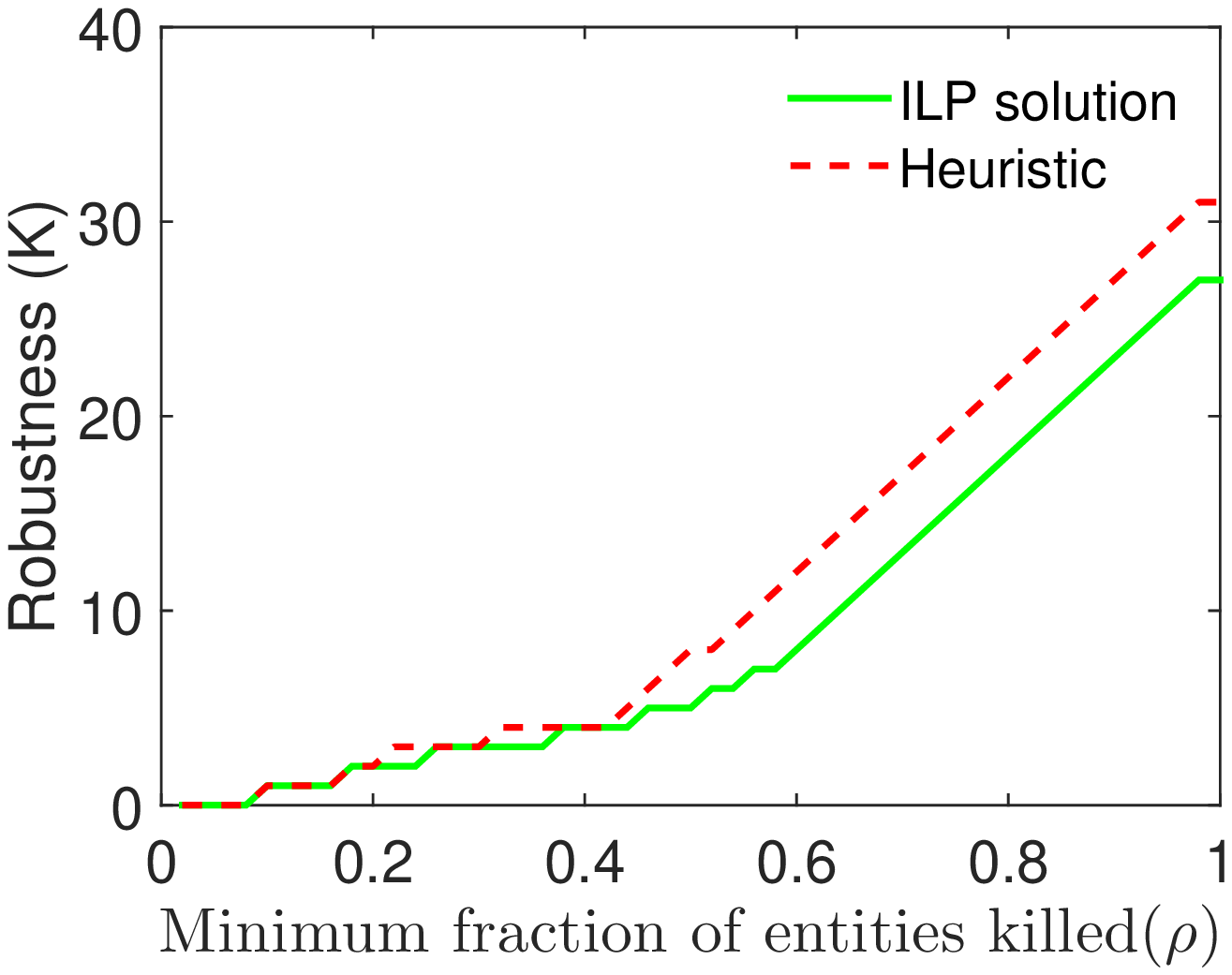}\label{fig:plot3}}
           \subfloat[][Region 4]{\includegraphics[width=5.5cm]{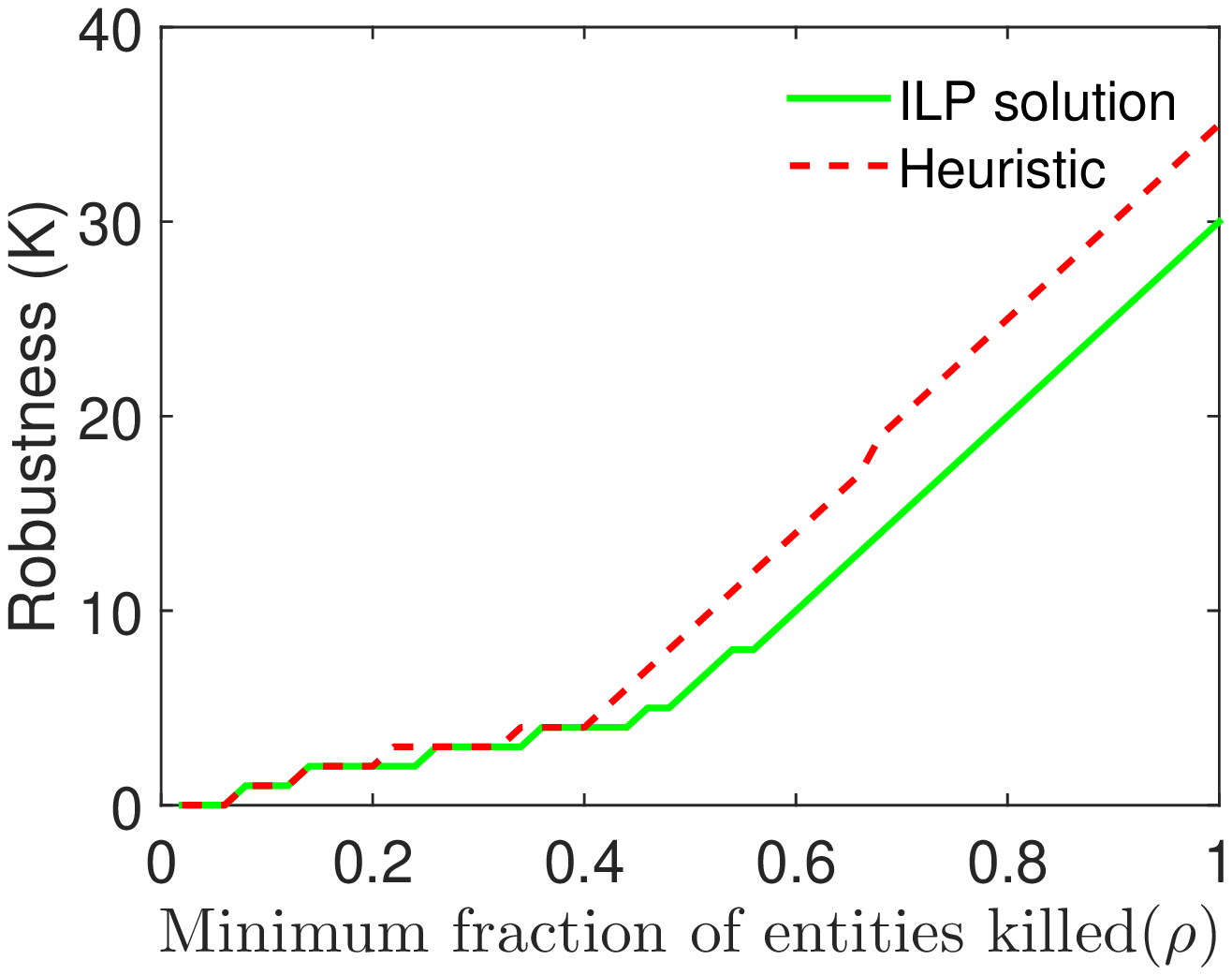}\label{fig:plot4}}
	\end{center}
	\caption{Robustness parameter $K$ returned by the optimal solution and the heuristic by varying parameter $\rho$ for four regions in Maricopa County, Arizona, USA}\label{fig:plot}
\vspace{-10pt}
\end{figure*}

IBM CPLEX Optimizer 12.5 is used to get the optimal solution using the ILP. The simulation for the heuristic was done in Python 3. In a given region the minimum fraction of entities killed ($\rho$) was varied from $0.02$ to $1$ in steps of $0.02$. For each value of $\rho$ the robustness parameter ($K$) was obtained from the optimal solution and the heuristic. Figures \ref{fig:plot1} to \ref{fig:plot4} shows the result obtained from the simulations for Region 1 to 4. It can be seen from the figures that the heuristic solution performs almost same as the ILP till $\rho=0.42$. For values of $\rho$ higher than $0.42$ there results in an overestimation of the robustness parameter $K$. The maximum overestimation is $3$ for Region 1, $4$ for Region 2 and 3 and $5$ for Region 4.

\section{Conclusion}
\label{Conc}
In this paper we propose the \textit{Robustness} problem in Multilayer Interdependent Network. \textit{Robustness} in an IDN is defined with respect to two parameters $K$ and $\rho$. The \textit{Implicative Interdependency Model} is utilized to model the interdependency. The IIM model segregates the interdependency relations into four different cases. Analysis of computational complexity of the \textit{Robustness} problem is done with respect to these four cases. For the general form of interdependency relation, the problem is found to be NP-complete. The optimal solution for the general case is obtained from an ILP. A heuristic is designed that returns an overestimated $K$ parameter value for a given $\rho$. Finally we compare the efficacy of the heuristic with the optimal solution using real data set of Maricopa County, Arizona. The heuristic produced optimal or near optimal solution for $\rho < 0.42$.

\vspace{-10pt}
\begin{footnotesize}
\bibliographystyle{splncs03}
\bibliography{references,referencesBibToAdd}
\end{footnotesize}
\end{document}